\documentclass[conference, letterpaper,twocolumn]{IEEEtran}
\usepackage{epsfig}
\usepackage{times}
\usepackage{float}
\usepackage{afterpage}
\usepackage{amsmath}
\usepackage{amstext}
\usepackage{amssymb,bm, bbm}
\usepackage{latexsym}
\usepackage{color}
\usepackage{graphicx}
\usepackage{amsmath}
\usepackage{amsthm}
\usepackage{graphicx}
\usepackage[center]{caption}
\usepackage{booktabs}
\usepackage{multicol}
\usepackage{lipsum}% dummy text
\usepackage{dblfloatfix}
\usepackage{mathrsfs}
\usepackage{cite}
\usepackage{tikz}
\usepackage{pgfplots}
\pgfplotsset{compat=newest}
 \usepackage{enumitem}
\usepackage{makecell}
\usepackage{float}
\restylefloat{table}
\usepackage{subcaption}

\usepackage{url}

\allowdisplaybreaks

\newcommand{\bbE}{\mathbb{E}}\newcommand{\rme}{\mathrm{e}}

\newcommand{\sfD}{\mathsf{D}}

\newcommand{\cA}{\mathcal{A}}

\newcommand{\sfd}{\mathsf{d}}

\newcommand{\cM}{\mathcal{M}}

\newcommand{\cP}{\mathcal{P}}

\newcommand{\cS}{\mathcal{S}}

\newcommand{\cX}{\mathcal{X}}

\theoremstyle{mystyle}
\newtheorem{theorem}{Theorem}%[section]
\theoremstyle{mystyle}
\newtheorem{lemma}{Lemma}%[section]
\theoremstyle{mystyle}
\newtheorem{prop}{Proposition}%[section]
\theoremstyle{mystyle}
\theoremstyle{mystyle}
\newtheorem{definition}{Definition}%[section]
\theoremstyle{remark}
\theoremstyle{mystyle}
\theoremstyle{mystyle}
\newtheorem{exa}{Example}%[section]
\theoremstyle{mystyle}
\theoremstyle{discussion}
\theoremstyle{mystyle}
\theoremstyle{mystyle}
 \usepackage{enumitem}

\begin{document}

\title{Lossy Source Coding with Focal Loss }
\author{
 \IEEEauthorblockN{Alex Dytso$^{\dagger}$ and Martina Cardone$^{*}$}
$^{\dagger}$ Qualcomm Flarion Technologies, Bridgewater,  NJ 08807, USA, odytso2@gmail.com\\
$^{*}$ University of Minnesota, Minneapolis, MN 55455, USA, mcardone@umn.edu
}

\maketitle

\begin{abstract}
Focal loss has recently gained significant popularity, particularly in tasks like object detection where it helps to address class imbalance by focusing more on hard-to-classify examples. This work 
%introduces 
proposes
the focal loss as a distortion measure for lossy source coding. The paper provides single-shot converse and achievability bounds. These bounds are then used to characterize the distortion-rate trade-off in the infinite blocklength, which is shown to be the same as that for the log loss case.  In the non-asymptotic case, the difference between focal loss and log loss is illustrated through a series of simulations.
\end{abstract}

\section{Introduction}
Recent advancements in the field of information theory have progressively embraced the concept of \emph{soft-reconstruction} - a paradigm wherein the decompressor outputs a probabilistic distribution rather than a discrete sample from the reconstruction alphabet. This approach is facilitated by the employment of the \emph{log loss} distortion function, more commonly recognized as \emph{cross-entropy} within the machine learning literature.

Concurrently, the machine learning community has enhanced model performance and convergence rates beyond traditional log loss methods by adopting alternative loss functions. Notably, the \emph{focal loss}, which originated in object detection, has gained widespread use across various fields. This function modifies the log loss by integrating a dynamic scaling factor that shifts training focus: it reduces emphasis on `easy examples' with acceptable error levels and increases focus on `hard examples' with significant errors.

In this work, we integrate the focal loss within a rate-distortion framework that employs soft-reconstruction. Specifically, we seek to compress a random variable $X \sim P_X$ with a finite alphabet $\mathcal{X}$. The reconstruction alphabet is given by $\cP(\cX)$, which denotes the set of all probability mass functions on $\cX$. We propose the following distortion measure:

\begin{definition}[Focal loss distortion]
\label{def:FocalLoss}
Given a \emph{weighting function} $\omega : \mathcal{X} \to \mathbb{R}_{+}$ and a \emph{focus parameter} $\gamma \geq 0$, the \emph{focal loss distortion} between an element $x \in \cX$ and its reconstruction $\hat{P} \in \cP(\mathcal{X})$ is defined as\footnote{Logarithms are assumed to be base two.}:
\begin{equation}
    \mathsf{d}(x; \hat{P}) = \underbrace{\omega(x)}_{\text{weight}} \underbrace{(1 - \hat{P}(x))^\gamma}_{\text{focal term}} \underbrace{\log \left( \frac{1}{\hat{P}(x)} \right)}_{\text{log loss}}.
\end{equation}
\end{definition}

The focal term $(1- \hat{P}(x))^\gamma$ includes a focus parameter $\gamma \ge 0$. As $\hat{P}(x) \to 1$, the loss for the symbol \(x\) becomes proportionately reduced. The focus parameter $\gamma$ controls the impact of the focal term. For $\gamma=0$, the focal loss reduces to the log loss. The term $\omega(x)$ adds additional information about the importance of the symbol $x \in \cX$. For average distortion, $\omega(x)$ can often be absorbed into the source distribution, but in Definition~\ref{def:FocalLoss} we retain it for completeness to illustrate the full structure of the focal loss.

\subsection{Literature Review}
The log loss and soft-reconstruction in the context of rate-distortion were introduced in~\cite{courtade2011multiterminal, courtade2013multiterminal}. Certain strong universality properties of the log loss were demonstrated in~\cite{no2015universality, BD-log_loss}. An axiomatic approach to justifying the log loss was undertaken in~\cite{jiao2015justification}. Single-shot bounds for lossy compression with the log loss were considered in~\cite{shkel2017single}. Error exponents for source coding with log loss were studied in~\cite{joudeh2024error}. In~\cite{wu2023soft}, the log loss was introduced to study a soft version of the Massey-Arikan guessing problem. The vector Gaussian CEO problem with the log loss was addressed in~\cite{uugur2020vector}. Connections between the log loss and the best quantizer based on the mutual information criterion were shown in~\cite{bhatt2021information}. Joint communication and sensing under the log loss for the recovery of the state was considered in~\cite{joudeh2024joint}. Joint compression and inference with log loss have been considered in~\cite{ulger2023single}. The log loss is also a special case of the \(\alpha\)-loss~\cite{liao2019tunable} studied under privacy settings; however, the focal loss coincides with $\alpha$-loss only for $\gamma=0$ and $\alpha=1$.

The focal loss was introduced in~\cite{lin2018focallossdenseobject} in the context of object detection and has since gained wide adaptation. Several generalizations of the focal loss have been proposed, e.g., modifying the focal term~\cite{li2022generalized}, making the focus parameter adaptive~\cite{ghosh2022adafocal}, or noting that 
%both 
the focal loss and the log loss are special cases of certain power series~\cite{leng2022polyloss}. For a recent survey on loss functions, the interested reader is referred to~\cite{ciampiconi2023survey}.

\subsection{Paper Contributions and Outline}
Section~\ref{sec:definitions} is dedicated to the problem formulation and relevant definitions.  Section~\ref{sec:converse} presents a single-shot converse bound and some other auxiliary results.  Section~\ref{sec:achiev} presents our achievable scheme.  Section~\ref{sec:examples} provides asymptotic bounds and some simulations.  
Finally, Section~\ref{sec:Conclusions} concludes the paper.

\section{Definitions and Single-Shot Preliminaries}
\label{sec:definitions}
In this section, we describe the single-shot approach and define its fundamental limit. 

Given a positive integer $M$ and a set $\mathcal{M} = \{1, \ldots, M \}$, a fixed-length source code of size $M$ is defined as a pair of functions $(f, g)$ such that 
\begin{align}
\text{Compressor: } f &: \mathcal{X} \to \mathcal{M}, \\
\text{Decompressor: } g &: \mathcal{M} \to \mathcal{P}(\mathcal{X}). 
\end{align}
A 
%lossy 
source code $(f, g)$ is an $(M, d)$-lossy source code if 
\begin{equation}
\label{eq:ExpectedDist}
\mathbb{E} \left[ \mathsf{d} (X; g(f(X))) \right] \le d \quad \text{and} \quad | \mathcal{M} | \le M,
\end{equation}
where $\mathsf{d}(\cdot ; \cdot)$ is the focal loss distortion with focus parameter $\gamma$. The fundamental limit of the single-shot approach is defined~as
\begin{equation}
\label{eq:FundLim}
d^*(M; \gamma) = \inf \{ d : (M, d)\text{-lossy source code} \}.
\end{equation}
As alluded to earlier, for the lossy source coding with average distortion\footnote{This may not be the case when considering the excess distortion.}, it is convenient to absorb the weight function $\omega(\cdot)$ into the source distribution $P_X$. In this work, we consider the parametrization of the weight function in terms of another probability distribution $Q \in \mathcal{P}(\mathcal{X})$ as follows:  
\begin{equation}
\omega(x) = \frac{Q(x)}{\sum_{a \in \mathcal{X}} P_X(a) Q(a)}, \quad x \in \mathcal{X},
\end{equation}
and define its re-weighted version with the source distribution $P_X$ as 
\begin{equation}
\label{eq:RX}
R_X(x) = P_X(x) \omega(x) = \frac{P_X(x) Q(x)}{\sum_{a \in \mathcal{X}} P_X(a) Q(a)}, \quad x \in \mathcal{X}.
\end{equation} 
All the results in this work can be stated with respect to $R_X$. Thus, moving forward, unless otherwise stated, the source distribution is given by $R_X$.  In particular, 
by absorbing $\omega(\cdot)$ into $P_X$ (i.e., by considering $R_X$ as the source distribution), the focal loss distortion in Definition~\ref{def:FocalLoss} can be simply defined~as
\begin{equation}
\label{eq:NewFocalLoss}
    \mathsf{d}(x; \hat{P}) =  (1 - \hat{P}(x))^\gamma\log \left( \frac{1}{\hat{P}(x)} \right),
\end{equation}
and the expected distortion in~\eqref{eq:ExpectedDist} is now given by 
\begin{align*}
&\mathbb{E} \left[ \mathsf{d} (X; g(f(X))) \right] \notag
\\& = \sum_{a \in \cX} R_X(a) \big( 1 - g(f(a)) \big)^\gamma \log \left( \frac{1}{g(f(a))} \right) .
\end{align*}
In the $n$-letter setting, we follow the approach set in~\cite{shkel2017single}: we consider the source $X^n$ over $\mathcal{X}^n$ and the reconstruction alphabet given by $\mathcal{P}(\mathcal{X}^n)$, i.e., the set of all probability distributions over $\mathcal{X}^n$. 
The $n$-letter distortion, after absorbing the weight function, is defined as
\begin{align}
\mathsf{d}_n(x^n; \hat{P}^n) & = \frac{\mathsf{d}(x^n; \hat{P}^n)}{n} \notag
\\& = \frac{1}{n} \big( 1 - \hat{P}^n(x^n) \big)^\gamma \log \left( \frac{1}{\hat{P}^n(x^n)} \right) , 
\end{align}
where $\hat{P}^n \in \mathcal{P}(\mathcal{X}^n)$. 
Similar to~\eqref{eq:FundLim}, we let $d_n^*(M;\gamma)$ denote the fundamental limit for the $n$-letter case. The \emph{distortion-rate function} is defined as 
\begin{equation}
\mathsf{D}(\mathsf{R}; \gamma) = \liminf_{n \to \infty} d_n^* \left( 2^{n\mathsf{R}}; \gamma \right), \quad \mathsf{R} > 0. 
\end{equation}

\section{Converse}
\label{sec:converse}
In this section, we present our converse bound. We start with some auxiliary results. 
\subsection{Auxiliary Results}
The following entropy-like quantity will be important in our derivation.
\begin{definition} Given a discrete random variable $X \sim P_X$ supported on $\cX$ and $\gamma \ge 0$, we let
\begin{equation}
\label{eq:EntropyTerm}
 H_\gamma(P_X) = \log \left( \sum_{x \in \cX} P_X(x)^{ ( 1 - P_X(x) )^\gamma } \right).
\end{equation}
\end{definition}
\begin{prop} 
The quantity $ H_\gamma(P_X)$ in~\eqref{eq:EntropyTerm} satisfies the following two basic properties:
\begin{itemize}[leftmargin=*]  
\item \emph{Monotonicity}:  $\gamma \mapsto   H_\gamma(P_X)$ is a non-decreasing function;  and
\item \emph{Non-negativity}: $H_\gamma(P_X) \ge  0$. 
\end{itemize}
\label{Prop:PropertyEntropy}
\end{prop}
\begin{IEEEproof}
The monotonicity in $\gamma$ follows from the fact that for any fixed $t \in (0,1)$ the function $t^{(1-t)^\gamma}$ is non-decreasing in $\gamma$. 
The non-negativity follows from the fact that
\begin{equation}
H_\gamma(P_X) \ge H_0(P_X) = 0,
\end{equation}
where we have used the monotonicity property. 
This concludes the proof of Proposition~\ref{Prop:PropertyEntropy}.
\end{IEEEproof}
We will also need a maximum entropy like principle for $H_\gamma(P_X)$ in~\eqref{eq:EntropyTerm}. To this end, we define 
\begin{equation}
h_\gamma(| \cX|) =   \max_{P_X : X \in \cX}  H_\gamma(P_X), \label{eq:Max_Entropy_Quantitiy}
\end{equation}
where  $|\cX|$ denotes the cardinality of $\cX$.
The key behavior that differentiates $H_\gamma(P_X)$ in~\eqref{eq:EntropyTerm} from the classical Shannon entropy is that it remains bounded as the cardinality of $\cX$  increases. This does not appear to be a new result{\color{blue},} and it indeed holds for quantities of the form
\begin{equation} \log \left( \sum_{x \in \cX} f(P_X(x)) \right),
\end{equation} where $f: [0,1] \to [0,1]$ is continuous and has a single inflection point on $[0,1]$, as does our function $f(t) = t^{(1-t)^\gamma}$. We found several sources providing proofs of this property~\cite[Prob.~21]{mildorf_inequalities},\cite{aops_problem_solving,math_stack_max_ent}. For completeness, a full proof of the following proposition is provided in Appendix~\ref{app:MaxEntrPrinc}.
%} 
%
\begin{prop}\label{prop:max_entropy} For any $| \cX |\ge 2$ and $\gamma \geq 0$, it holds that\footnote{The notation $[a \!:\! b]$ represents the set of integers from $a$ to $b $ where $ b > a$.}
\begin{align}
    h_\gamma(| \cX|) &= \!\!\!\! \max_{ \small \begin{array}{c} d \in [1\!:\! | \cX|\!-\!1] \\  q \in [0,1]  \end{array}} \!\!\log \left( (1-q)^{q^\gamma} + d \left( \frac{q}{d} \right)^{ \left( 1 - \frac{q}{d} \right)^\gamma } \right) \label{eq:FirstClaimPropMaxEntr}\\
& \le \log \left( 1 + \rme^{ \frac{\max(1,\gamma) }{\rme}}  \right).
\label{eq:SecondClaimPropMaxEntr}
\end{align} 
\end{prop}

Fig.~\ref{fig:max_entropy} shows the behavior of $h_\gamma(| \cX|)$ versus $| \cX|$.   
\begin{figure}
\centering
\definecolor{mycolor1}{rgb}{0.00000,0.44700,0.74100}%
\definecolor{mycolor2}{rgb}{0.85000,0.32500,0.09800}%
\definecolor{mycolor3}{rgb}{0.92900,0.69400,0.12500}%
\definecolor{mycolor4}{rgb}{0.49400,0.18400,0.55600}%
\begin{tikzpicture}

\begin{axis}[%
width=6cm,
height=3.5cm,
at={(1.011in,0.642in)},
scale only axis,
xmin=1,
xmax=50,
xlabel style={font=\color{white!15!black}},
xlabel={$| \cX|$},
ymin=0,
ymax=5,
ylabel = {$h_\gamma(| \cX|)$},
axis background/.style={fill=white},
xmajorgrids,
ymajorgrids,
legend style={legend cell align=left, align=left, draw=white!15!black, at={(0.9,0.57)}, font = \small }
]
\addplot [color=black, thick, solid]
  table[row sep=crcr]{%
2	0.292893218813452\\
3	0.292893218813452\\
4	0.292893218813452\\
5	0.292893218813452\\
6	0.292893218813452\\
7	0.292893218813452\\
8	0.292893218813452\\
9	0.292893218813452\\
10	0.292893218813452\\
11	0.292893218813452\\
12	0.292893218813452\\
13	0.292893218813452\\
14	0.292893218813452\\
15	0.292893218813452\\
16	0.292893218813452\\
17	0.292893218813452\\
18	0.292893218813452\\
19	0.292893218813452\\
20	0.292893218813452\\
21	0.292893218813452\\
22	0.292893218813452\\
23	0.292893218813452\\
24	0.292893218813452\\
25	0.292893218813452\\
26	0.292893218813452\\
27	0.292893218813452\\
28	0.292893218813452\\
29	0.292893218813452\\
30	0.292893218813452\\
31	0.292893218813452\\
32	0.292893218813452\\
33	0.292893218813452\\
34	0.292893218813452\\
35	0.292893218813452\\
36	0.292893218813452\\
37	0.292893218813452\\
38	0.292893218813452\\
39	0.292893218813452\\
40	0.292893218813452\\
41	0.292893218813452\\
42	0.292893218813452\\
43	0.292893218813452\\
44	0.292893218813452\\
45	0.292893218813452\\
46	0.292893218813452\\
47	0.292893218813452\\
48	0.292893218813452\\
49	0.292893218813452\\
50	0.292893218813452\\
};
\addlegendentry{$\gamma =0.5$}

\addplot [color=black, thick, dashed]
  table[row sep=crcr]{%
2	0.5\\
3	0.52832083343362\\
4	0.52832083343362\\
5	0.52832083343362\\
6	0.52832083343362\\
7	0.52832083343362\\
8	0.52832083343362\\
9	0.52832083343362\\
10	0.52832083343362\\
11	0.52832083343362\\
12	0.52832083343362\\
13	0.52832083343362\\
14	0.52832083343362\\
15	0.52832083343362\\
16	0.52832083343362\\
17	0.52832083343362\\
18	0.52832083343362\\
19	0.52832083343362\\
20	0.52832083343362\\
21	0.52832083343362\\
22	0.52832083343362\\
23	0.52832083343362\\
24	0.52832083343362\\
25	0.52832083343362\\
26	0.52832083343362\\
27	0.52832083343362\\
28	0.52832083343362\\
29	0.52832083343362\\
30	0.52832083343362\\
31	0.52832083343362\\
32	0.52832083343362\\
33	0.52832083343362\\
34	0.52832083343362\\
35	0.52832083343362\\
36	0.52832083343362\\
37	0.52832083343362\\
38	0.52832083343362\\
39	0.52832083343362\\
40	0.52832083343362\\
41	0.52832083343362\\
42	0.52832083343362\\
43	0.52832083343362\\
44	0.52832083343362\\
45	0.52832083343362\\
46	0.52832083343362\\
47	0.52832083343362\\
48	0.52832083343362\\
49	0.52832083343362\\
50	0.52832083343362\\
};
\addlegendentry{$\gamma =1$}

\addplot [color=black, thick, dotted]
  table[row sep=crcr]{%
2	0.999999046325684\\
3	1.5844858570487\\
4	1.99365757612211\\
5	2.29515808630867\\
6	2.51753620010795\\
7	2.67871921748462\\
8	2.79237372365946\\
9	2.86931838078094\\
10	2.9180591881791\\
11	2.94520915187907\\
12	2.95587160589104\\
13	2.95587160589104\\
14	2.95587160589104\\
15	2.95587160589104\\
16	2.95587160589104\\
17	2.95587160589104\\
18	2.95587160589104\\
19	2.95587160589104\\
20	2.95587160589104\\
21	2.95587160589104\\
22	2.95587160589104\\
23	2.95587160589104\\
24	2.95587160589104\\
25	2.95587160589104\\
26	2.95587160589104\\
27	2.95587160589104\\
28	2.95587160589104\\
29	2.95587160589104\\
30	2.95587160589104\\
31	2.95587160589104\\
32	2.95587160589104\\
33	2.95587160589104\\
34	2.95587160589104\\
35	2.95587160589104\\
36	2.95587160589104\\
37	2.95587160589104\\
38	2.95587160589104\\
39	2.95587160589104\\
40	2.95587160589104\\
41	2.95587160589104\\
42	2.95587160589104\\
43	2.95587160589104\\
44	2.95587160589104\\
45	2.95587160589104\\
46	2.95587160589104\\
47	2.95587160589104\\
48	2.95587160589104\\
49	2.95587160589104\\
50	2.95587160589104\\
};
\addlegendentry{$\gamma=20$}

\addplot [color=black, thick, dash dot]
  table[row sep=crcr]{%
2	1\\
3	1.58496250072116\\
4	1.99999999999936\\
5	2.32192809441438\\
6	2.58496246950857\\
7	2.80735435501138\\
8	2.99999523649527\\
9	3.16990069077894\\
10	3.32183985982713\\
11	3.45918058249838\\
12	3.58436597582024\\
13	3.69920364907841\\
14	3.80505265827097\\
15	3.90295048047733\\
16	3.99370221738772\\
17	4.07794452864719\\
18	4.15619186066929\\
19	4.22886994294639\\
20	4.29633999327239\\
21	4.35891608869335\\
22	4.41687756856133\\
23	4.47047778510773\\
24	4.51995024434857\\
25	4.56551285273497\\
26	4.60737081992397\\
27	4.64571867930445\\
28	4.68074157493267\\
29	4.71261622817619\\
30	4.7415114676573\\
31	4.76758871641806\\
32	4.79100226852174\\
33	4.81189952374772\\
34	4.83042116978648\\
35	4.84670138239203\\
36	4.86086798942739\\
37	4.87304264124267\\
38	4.88334102008057\\
39	4.8918729728088\\
40	4.89874278706781\\
41	4.90404931626217\\
42	4.90788624971978\\
43	4.91034225935592\\
44	4.91150127694133\\
45	4.91150127694133\\
46	4.91150127694133\\
47	4.91150127694133\\
48	4.91150127694133\\
49	4.91150127694133\\
50	4.91150127694133\\
};
\addlegendentry{$\gamma = 100$}

\end{axis}

\end{tikzpicture}%
\caption{Examples of $h_\gamma(| \cX|)$ in \eqref{eq:FirstClaimPropMaxEntr}.}
\label{fig:max_entropy}
\end{figure}
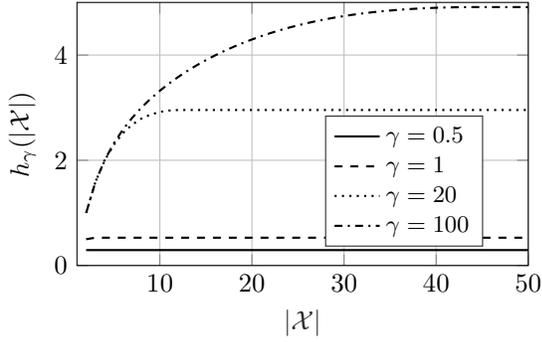

\subsection{Converse Bound} 
\begin{theorem}
\label{thm:ConverseBound}
It holds that
\begin{equation}
d^*(M;\gamma)  \ge   \left [ H(X) - \log(M)-h_\gamma(| \cX|)\right ]^+, \label{eq:converse}
\end{equation}
where $X \sim R_X(x)$ with $R_X$ being defined in~\eqref{eq:RX}. 
\end{theorem}

\begin{proof}
Fix some $u \in \cM$,  let $\hat{P}_u = g(u)$, and define 
\begin{align}
\tilde{P}_u(x) &=  \left( \hat{P}_u(x) \right)^{ \left(1 -\hat{P}_u(x)  \right)^\gamma },\\
c(u) &= \sum_{x \in \cX } \tilde{P}_u(x). \label{eq:cu}
\end{align}
Then, for $(X, U) \sim R_X P_{U|X}$, we have that
\begin{align}
&\bbE \left[ \sfd( X; \hat{P}_U(X)) |U =u \right] \notag\\
&=\sum_{x \in \cX} {R}_{X|U}(x|u) \ \sfd \left(x; \hat{P}_u(x) \right) \notag \\
&= \sum_{x \in \cX} R_{X|U}(x|u)  
\left(1 -\hat{P}_u(x)  \right)^\gamma \log \left( \frac{1}{\hat{P}_u(x)} \right) \notag \\
&= \sum_{x \in \cX} R_{X|U}(x|u) 
\log \left( \frac{1}{\tilde{P}_u(x)} \right) \notag \\
& \stackrel{{\rm{(a)}}}{=} \sum_{x \in \cX}  R_{X|U}(x|u)  \log \left( \frac{R_{X|U}(x|u) }{\tilde{P}_u(x)} \right) + H \left( X| U=u\right) \notag \\
& \stackrel{{\rm{(b)}}}{=} D \left( R_{X|U}(\cdot|u)  \left \|  \frac{\tilde{P}_u}{c(u)} \right. \right ) + H \left( X |U =u \right) + \log \left( \frac{1}{c(u)} \right) \notag \\
& \stackrel{{\rm{(c)}}}{ \ge } H \left( X|U=u \right) + \log \left( \frac{1}{c(u)} \right) \notag \\
& \stackrel{{\rm{(d)}}}{=} H \left( X|U=u \right) - H_\gamma (\hat{P}_u(\cdot)) \notag \\
&    \stackrel{{\rm{(e)}}}{\ge} H \left( X|U=u \right) -   h_\gamma(| \cX|) ,
\label{eq:ConvStep2}
\end{align}
where the labeled (in)equalities follow from:
$\rm{(a)}$ multiplying and dividing the argument of the logarithm by $R_{X|U}(x|u)$;
$\rm{(b)}$ multiplying and dividing the argument of the logarithm by $c(u)$ in~\eqref{eq:cu};
 $\rm{(c)}$ the non-negativity of the relative entropy; $\rm{(d)}$ the definition of $H_\gamma $ in~\eqref{eq:EntropyTerm}; and  $\rm{(e)}$ the definition in~\eqref{eq:Max_Entropy_Quantitiy}.
 
Now, averaging \eqref{eq:ConvStep2} over $U \sim P_U$, we arrive at 
\begin{align}
\bbE \left[\sfd \left(X; g(U) \right)  \right] & = \sum_{u \in \cM}P_U(u) \ \bbE \left[\sfd \left(X; g(U) \right) |U=u \right] \notag
\\& \geq \sum_{u \in \cM}P_{U}(u) \left( H(X | U \!=\!u) \! - h_\gamma(| \cX|) \right) \notag
\\& = H(X| U) - h_\gamma(| \cX|) \notag\\
&=H(X) - I( U; X)- h_\gamma(| \cX|)  \notag \\
&\ge  H(X) - \log(M)-h_\gamma(| \cX|).
\end{align}
The proof of Theorem~\ref{thm:ConverseBound} is concluded by fixing an arbitrary $(M,d)$-lossy source code and by letting $U=f(X)$ be the compressor output.
\end{proof}

\section{Achievability}
\label{sec:achiev}

We now present our achievability bounds, which are an adaption of the achievability bounds in~\cite{shkel2017single} for the log loss. In addition, to account for the focus parameter $\gamma$, we also introduce an auxiliary distribution, denoted as $F_X$, over which we could optimize the bound. 

\begin{theorem}
\label{thm:AchievabilityBound}
If $|\cX|\leq M$, then $d^*(M;\gamma) = 0$. Otherwise, 
it holds that
\begin{align}
d^*(M;\gamma) & < \bbE_{R_X} \left [1_\cA \times \left( \frac{2^{\iota_{F_X} (X) - \log(M)}}{2^{\iota_{F_X} (X)-\log(M)}+1}\right )^\gamma  \right . \notag
\\& \qquad \qquad \left. \log \left ( 1+ 2^{\iota_{F_X} (X) - \log(M)}\right ) \right ] \label{eq:Second_Last_achiev}
\\& <\bbE_{R_X} \left [1_\cA \times  \left(1- \frac{1}{2} 2^{\log(M) - \iota_{F_X} (X)} \right )^{\gamma}  \right. \notag
\\& \qquad \qquad \left. \times \left( \iota_{F_X} (X)- \log \left( M \right) +1 \right )\right ], \label{eq:Last_achiev}
\end{align}
where: $R_X$ is defined in~\eqref{eq:RX}; $F_X$ is an arbitrary probability mass function on $\cX$;  $\cA = \{ \iota_{F_X} (X) > \log(M)\}$; and 
\begin{equation}
\label{eq:Information}
\iota_{F_X}(x) = \log \left(\frac{1}{F_X(x)} \right ). 
\end{equation}
\end{theorem}
\begin{proof}
Note that if $|\cX|\leq M$, then $d^*(M;\gamma) = 0$. This is because, we can select $f(a)=a$ and $g(a)=\delta_a$, where $\delta_a$ is the point mass at $a$. Then, for every $a \in \cX$, we have that
\begin{align} 
\sfd(x; g(f(a))) &= \sfd(x; g(a)) = (1-\delta_a)^\gamma \log \frac{1}{\delta_a} = 0.
\end{align} 
Thus, we focus on the case $|\cX| > M$.

\smallskip
\noindent
{\em Compressor.} Without loss of generality, we let $\cX = \{1,\ldots,|\cX|\}$, and $F_X(a) \geq F_X(b)$ for $a \leq b$.
We define the compressor as follows:
\begin{subequations}
\label{eq:Encoder}
\begin{equation}
f(a) = a \ \text{for} \ a \in [1:M],
\end{equation}
and 
\begin{equation}
 f(a) = \arg \min_{m \in \cM} \sum_{b=1}^{a-1} F_X(b) 1 \{f(b)=m\}, 
 \end{equation}
for $a \in [M+1:|\cX|]$.
\end{subequations}

\smallskip
\noindent
{\em Decompressor.} The decompressor for $m \in \cM$ is defined by $g(m) = \hat{P}_m$, where
\begin{equation}
\hat{P}_m(a) = \left \{
\begin{array}{ll}
\frac{F_X(a)}{\sum_{x \in \cX} F_X(x)  1\{f(x)=m\}} & f(a)=m,
\\
0 & \text{otherwise.}
\end{array}
\right .
\end{equation}

\smallskip
\noindent
{\em Average distortion analysis.}
The exact average distortion of $(f,g)$ is given by
\begin{align}
&\bbE \left [ \sfd(X; g(f(X)))\right ] \notag
%
%\\& {\color{red}= \sum_{a \in \cX} P_X(a) \ \sfd(a;g(f(a)))} \notag 
%
\\ &= \sum_{a \in \cX}R_X(a)  \left( 1-g(f(a))\right )^\gamma \log \left( \frac{1}{g(f(a))} \right ) \notag
\\& = \sum_{a \in \cX} R_X(a) \log \left(\frac{\sum_{x \in \cX} F_X(x)  1\{f(x)=f(a)\}}{F_X(a)} \right )\notag
\\ & \qquad \quad \left(1- \frac{F_X(a)}{\sum_{x \in \cX} F_X(x)  1\{f(x)=f(a)\}} \right )^\gamma \notag
\\& = \sum_{a \in \cX} R_X(a) \log \left( \frac{\mathbb{P}_{F_X} \left (f(X)=f(a) \right )}{F_X(a)} \right )
\notag
\\ & \qquad \quad \left(1- \frac{F_X(a)}{\mathbb{P}_{F_X} \left (f(X)=f(a) \right )} \right )^\gamma, \label{eq:Exect_Achiev}
\end{align}
where, in the last step, $\mathbb{P}_{F_X}$ means that the probability is computed with respect to $F_X$.

We now want to upper bound the above expression. Towards this end, assume first that $F_X(a) \geq 1/M$.
Under this assumption, we claim that $a$ is the only element assigned to $f(a)$.
To see this, note the following two facts:
(1) from the compressor, we have that $f(a)=a$ for any $a \in [1:M]$, i.e., we have that $a$ is the only element assigned to $f(a)$ for any $a \in [1:M]$;
and (2) for any $a \in [M+1:|\cX|]$ we have that
\begin{align}
1 = \sum_{x \in \cX} F_X(x) & = \sum_{i=1}^{a-1} F_X(i) + F_X(a) + \sum_{i=a+1}^{|\cX|} F_X(i) \notag
\\& \stackrel{{\rm{(a)}}}{\geq} a F_X(a) + \sum_{i=a+1}^{|\cX|} F_X(i) \notag
\\& \stackrel{{\rm{(b)}}}{\geq} (M+1) \frac{1}{M} + \sum_{i=a+1}^{|\cX|} F_X(i) >1, \label{eq:Contradiction}
\end{align}
where $\rm{(a)}$ follows since we assume that (without loss of generality) $F_X(a)\geq F_X(b)$ for $a \leq b$, and $\rm{(b)}$ is due to the fact that, by assumption, $a \geq M+1$ and $F_X(a)\geq 1/M$.
Thus, from~\eqref{eq:Contradiction} we have a contradiction, i.e., under the assumption $F_X(a)\geq 1/M$, we need to have $a\leq M$.
It therefore follows that, if $F_X(a) \geq 1/M$, then $a$ is the only element assigned to $f(a)$. Hence, under this assumption, we arrive at
\begin{equation}
\sfd(a;g(f(a))) = \sfd(a;g(a)) = 0.
\end{equation}
We now assume $F_X(a) < 1/M${\color{blue},} and we let $b \in \cX$ be the last element in $\cX$ assigned to $m=f(a)$. For this case, we have that        
\begin{align}
&\mathbb{P}_{F_X} \left [f(X)=m \right ] \notag
\\& 
=\sum_{i=1}^{b-1} F_X(i) 1 \{f(i)=m\}  + F_X(b) \label{eq:PfXmInt}
\\& \leq \min_{\substack{j \in [1:M]\\ j \neq m}} \left \{\sum_{i=1}^{b-1} F_X(i) 1 \{f(i)=j\} \right \}  + F_X(b), \label{eq:PfXm}
\end{align}
where the inequality follows from the fact that $m$ is a minimum of~\eqref{eq:Encoder}.
Now, it is not difficult to see that $\sum_{i=1}^{b-1} F_X(i) 1 \{f(i)=m\}$ in~\eqref{eq:PfXmInt} has to be smaller than $1/M$. To see this, consider the opposite case, i.e.,  assume that $\sum_{i=1}^{b-1} F_X(i) 1 \{f(i)=m\} \geq 1/M$. Then, this implies that each of the $M-1$ terms in the minimum in~\eqref{eq:PfXm} is also greater than $1/M$. Thus, we would get
\begin{align}
&\sum_{i=1}^{b-1} F_X(i) 1 \{f(i)=m\} 
\!+\!\!\!\!\!\!\! \sum_{j\in [1:M], j \neq m} \left \{\sum_{i=1}^{b-1} F_X(i) 1 \{f(i)\!=\!j\} \right \} \notag
\\& \geq \frac{1}{M} + (M-1)\frac{1}{M} =1.
\end{align}
This creates a contradiction since it would lead to $F_X(b)=0$ (since $F_X$ is a probability mass function), which is not possible since it was assumed that $b \in \cX$. It therefore follows that 
$\sum_{i=1}^{b-1} F_X(i) 1 \{f(i)=m\} < 1/M$. With this, from~\eqref{eq:PfXmInt}, 
%{\color{blue} (AD: should it be (31) instead of (32) )} 
we arrive at
\begin{align}
\mathbb{P}_{F_X} \left [f(X)=m \right ] &= \sum_{i=1}^{b-1} F_X(i) 1 \{f(i)=m\}  + F_X(b) \notag
\\& < \frac{1}{M} + F_X(b) \notag
\\& \leq \frac{1}{M} + F_X(a) \label{eq:UBPFxInt}
\\& < \frac{2}{M},
\label{eq:UBPFx}
\end{align}
where the second inequality follows since (without loss of generality)
we assume that $F_X(a) \geq F_X(b)$ for $a \leq b$, and the last inequality is due to the fact that we are considering $F_X(a)<1/M$.  
Thus, for all $a$ for which $F_X(a) < 1/M$, it holds that
\begin{align}
&R_X(a) \log \left( \frac{\mathbb{P}_{F_X} \left (f(X)=f(a) \right )}{F_X(a)} \right ) \notag
\\& \qquad \quad  \times \left(1- \frac{F_X(a)}{\mathbb{P}_{F_X} \left (f(X)=f(a) \right )} \right )^\gamma \notag
\\& \stackrel{\eqref{eq:UBPFxInt}}{<} R_X(a) \log \left (1\!+\!\frac{1}{M F_X(a)} \right )\left( 1\!-\!\frac{M F_X(a)}{1\!+\! MF_X(a)}\right )^\gamma
\\& \stackrel{\eqref{eq:UBPFx}}{<} R_X(a) \log \left( \frac{2/M}{F_X(a)} \right ) \left(1-\frac{F_X(a)}{2/M} \right )^\gamma.
\end{align}
We, therefore, arrive at
\begin{align}
& d^*(M;\gamma)  \notag
\\ & < \!\bbE_{R_X} \!\!\left [\!1_\cA \! \log \!\left (1\!+\!\frac{1}{M F_X(X)} \!\right )\!\!\left( 1\!-\!\frac{M F_X(X)}{1\!+\! MF_X(X)}\!\right )^\gamma \right ]
\\ & <\bbE_{R_X} \left [1_\cA  \log \left( \frac{2/M}{F_X(X)} \right ) \left(1-\frac{F_X(X)}{2/M} \right )^\gamma  \right ],
\end{align}
where we let
\begin{equation}
\cA = \{ F_X(X)<1/M\}. 
%\notag \\
%&=\{ \iota_{R} (X) > \log(M)\}.
\end{equation}
The proof of Theorem~\ref{thm:AchievabilityBound} is concluded by using~\eqref{eq:Information} in the above. 
\end{proof}

\section{Asymptotics, Comparisons, and Evaluations}
\label{sec:examples}

\subsection{Asymptotics}
We now discuss the asymptotic fundamental limits for a stationary memoryless source.  
\begin{theorem}\label{thm:assymptotics}
Suppose that $R_{X^n}= \prod_{i=1}^n R_X(x_i)$ for all $n$. Then,  for  $\gamma \ge 0$ and $\mathsf{R} \ge0 $
\begin{equation}
 \sfD (\mathsf{R};\gamma) = [ H(X) - \mathsf{R} ]^{+},    
\end{equation}
where $X \sim R_X$. 
\end{theorem}
\begin{IEEEproof}
    First, from the converse bound in Theorem~\ref{thm:ConverseBound} we have that
    \begin{align}
    \sfD (\mathsf{R};\gamma) &\ge \liminf_{n \to \infty} \frac{\left [ H(X^n) - \log(2^{n\mathsf{R}})-h_\gamma(| \cX^n|)\right ]^+}{n} \notag \\
    &= [ H(X) - \mathsf{R}]^{+},
    \end{align}
    where the fact that $\liminf_{n \to \infty} \frac{h_\gamma(| \cX^n|)}{n}=0$ follows from Proposition~\ref{prop:max_entropy}. 

For the upper bound  in Theorem~\ref{thm:AchievabilityBound}, we can choose $F_{X^n}(x^n) = \prod_{i=1}^n F_{X}(x_i)$ and observe that 
    \begin{align}
    & \liminf_{n \to \infty}  \frac{1}{n}\bbE_{R_{X^n} } \left [1_\cA \times  \left(1- \frac{1}{2} 2^{\log(2^{n\mathsf{R}}) - \iota_{F_{X^n}  } (X^n)} \right )^{\gamma}  \right. \notag
\\& \qquad \qquad \left. \times \left( \iota_{F_{X^n}} (X^n)- \log \left( 2^{n\mathsf{R}} \right) +1 \right )\right ] \notag \\
  & \le  \liminf_{n \to \infty}  \frac{1}{n}\bbE_{R_{X^n} } \left [1_\cA      \times \left( \iota_{F_{X^n}} (X^n)- n\mathsf{R} +1 \right )\right ] \notag\\
  &= \left[ \bbE_{R_{X} } \left [      \iota_{F_{X}} (X) \right ]- \mathsf{R}    \right]^{+} ,
    \end{align}
    where the inequality follows from bounding the focal term with $\left(1- \frac{1}{2} 2^{\log(2^{n\mathsf{R}}) - \iota_{F_{X^n}  } (X^n)} \right )^{\gamma} \le 1$, and the last equality follows from the law of large numbers. 
To complete the proof note that the optimal choice is $F_X= R_X$; this is due to Gibbs' inequality, that is,
    \begin{equation}
        \bbE_{R_{X} } \left [      \iota_{F_{X}} (X) \right ] \ge \bbE_{R_{X} } \left [      \iota_{R_{X}} (X) \right ].
    \end{equation}
    This concludes the proof of Theorem~\ref{thm:assymptotics}.
\end{IEEEproof}

From Theorem~\ref{thm:assymptotics}, we make the following observations. First, the focal term has no impact asymptotically. This, as it will be shown in the next subsection, is not true in the finite blocklength regime. Second, asymptotically, the optimal choice of $F_X$ is to set it equal to the re-weighted source distribution $R_X$. As it will also be shown next, this is not necessarily true in the finite blocklength regime.

\subsection{Examples}
We here produce several examples to illustrate the impact of the focal term in the finite blocklength regime. 

\begin{exa} \label{examp:M2X3} Case $M=2$ and $|\cX| =3$.     For this case the average distortion can be accurately computed with an exhaustive search.  Fig.~\ref{fig:example_M2_X3} demonstrates two examples of $d^*(2;\gamma)$. 
\end{exa}

\begin{figure}[h]
\centering
\begin{tikzpicture}

\begin{axis}[%
width=6cm,
height=3.4cm,
at={(1.011in,0.642in)},
scale only axis,
xmin=0,
xmax=10,
xlabel style={font=\color{white!15!black}},
xlabel={$\gamma$},
ymin=0,
ymax=0.8,
ylabel = {$d^*(2;\gamma)$},
axis background/.style={fill=white},
xmajorgrids,
ymajorgrids,
legend style={legend cell align=left, align=left, draw=white!15!black}
]
\addplot [color=black, thick]
  table[row sep=crcr]{%
0	                0.6667\\
0.526315789473684	0.4629\\
1.05263157894737	0.3214\\
1.57894736842105	0.2232\\
2.10526315789474	0.1549\\
2.63157894736842	0.1076\\
3.15789473684211	0.0747\\
3.68421052631579	0.0519\\
4.21052631578947	0.0360\\
4.73684210526316	0.0250\\
5.26315789473684	0.0174\\
5.78947368421053	0.0121\\
6.31578947368421	0.0084\\
6.84210526315789	0.0058\\
7.36842105263158	0.0040\\
7.89473684210526	0.0028\\
8.42105263157895	0.0019 \\
8.94736842105263	0.0014\\
9.47368421052632	0.0009\\
10	0.0007\\
};
\addlegendentry{$R_{X_1}$}

\addplot [color=black, thick, dashed]
  table[row sep=crcr]{%
0	0.270426041486378\\
0.526315789473684	0.193270123296634\\
1.05263157894737	0.147236206978196\\
1.57894736842105	0.111575341768706\\
2.10526315789474	0.0774696129173672\\
2.63157894736842	0.0537891332478085\\
3.15789473684211	0.0373471706724107\\
3.68421052631579	0.0259310956138338\\
4.21052631578947	0.0180046227767001\\
4.73684210526316	0.0125010699956052\\
5.26315789473684	0.0086798125666515\\
5.78947368421053	0.00602661581918083\\
6.31578947368421	0.00418443347170251\\
6.84210526315789	0.00290535916083735\\
7.36842105263158	0.00201726515920138\\
7.89473684210526	0.00140063878414086\\
8.42105263157895	0.000972499323993806\\
8.94736842105263	0.000675231148727989\\
9.47368421052632	0.000468830253105065\\
10	0.000325520833333333\\
};
\addlegendentry{$R_{X_2}$}

\end{axis}
\end{tikzpicture}%
\caption{Exact $d^*(2; \gamma)$ for Example~\ref{examp:M2X3}, where  $R_{X_1} = [1/3, 1/3,1/3]$ and $R_{X_2} = [2/3, 1/4, 1/12]$. }
\label{fig:example_M2_X3}
\end{figure}
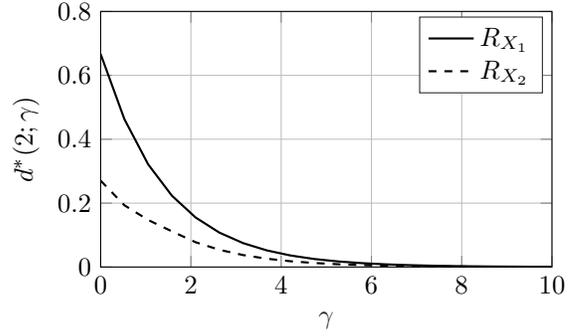

We next try a source distribution with a larger alphabet size by adopting~\cite[Example~3]{shkel2017single}.

\begin{exa}\label{examp:binomial} Let $X$ be a $\text{Binomial}(k, p)$ source
with $k = 100$. Fig.~\ref{fig:example_all_gammas} compares the converse in~\eqref{eq:converse} and the achievability bounds in~\eqref{eq:Exect_Achiev},~\eqref{eq:Last_achiev}, and~\eqref{eq:Second_Last_achiev} for $p=0.5$ and %$\gamma=1$
$\gamma \in  \{0.1, 1, 2, 4\}$
and it varies $M$. Fig.~\ref{fig:example_M8_vary_g} does a similar comparison for $p=0.1, M=8$, and it varies~$\gamma$.
\end{exa}

\begin{figure}[h!]
\centering

\begin{subfigure}{0.45\textwidth}
\centering
\input{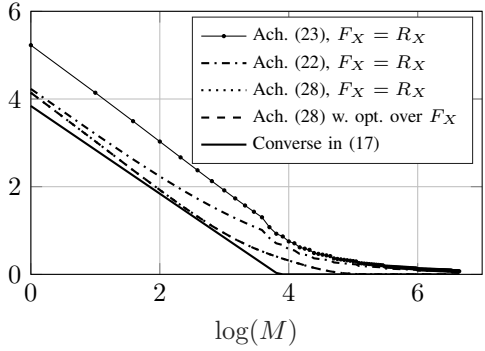}
\caption{$\gamma=1$}
\label{fig:example_g1_vary_M}
\end{subfigure}
\hfill
\begin{subfigure}{0.45\textwidth}
\centering
\input{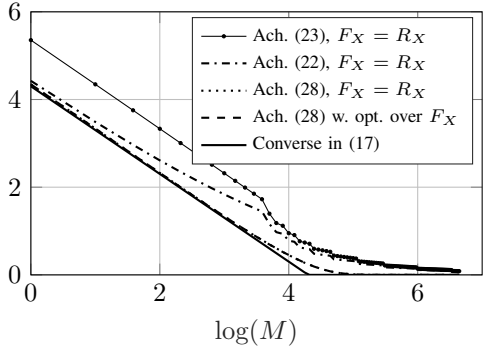}
\caption{$\gamma=0.1$}
\label{fig:example_g01_vary_M}
\end{subfigure}

\vspace{1em}

\begin{subfigure}{0.45\textwidth}
\centering
\input{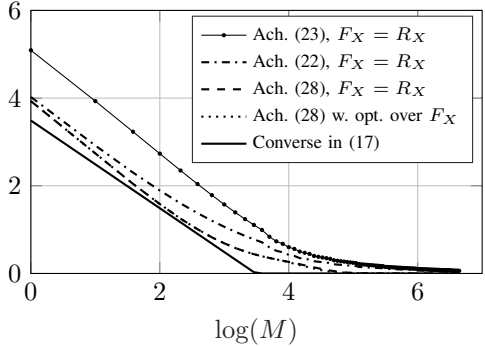}
\caption{$\gamma=2$}
\label{fig:example_g2_vary_M}
\end{subfigure}
\hfill
\begin{subfigure}{0.45\textwidth}
\centering
\input{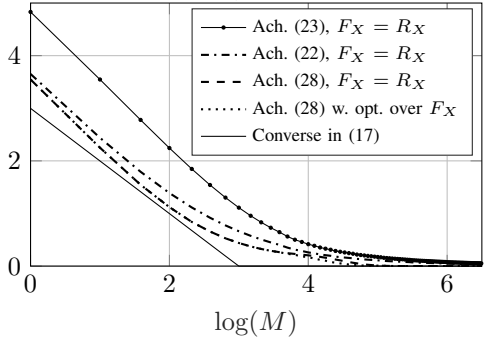}
\caption{$\gamma=4$}
\label{fig:example_g4_vary_M}
\end{subfigure}

\caption{Converse and achievability bounds for the binomial source in Example~\ref{examp:binomial} for different values of $\gamma$.}
\label{fig:example_all_gammas}
\end{figure}

\begin{figure}[h]
\centering
\definecolor{mycolor1}{rgb}{1.00000,0.00000,1.00000}%
\definecolor{mycolor2}{rgb}{0.00000,1.00000,1.00000}%
\begin{tikzpicture}

\begin{axis}[%
width=6cm,
height=3.4cm,
at={(1.011in,0.642in)},
scale only axis,
xmin=0,
xmax=10,
xlabel style={font=\color{white!15!black}},
xlabel={$\gamma$},
ymin=0,
ymax=2,
axis background/.style={fill=white},
xmajorgrids,
ymajorgrids,
legend style={legend cell align=left, align=left, draw=white!15!black, font =\scriptsize}
]
\addplot [color=black, mark=*, mark size=0.6pt ]
  table[row sep=crcr]{%
0	1.86897353302468\\
0.256410256410256	1.72520148068136\\
0.512820512820513	1.5964368288516\\
0.769230769230769	1.48093490032203\\
1.02564102564103	1.37716494074111\\
1.28205128205128	1.28378305481447\\
1.53846153846154	1.19960863706492\\
1.79487179487179	1.12360383941749\\
2.05128205128205	1.05485567846468\\
2.30769230769231	0.992560437775296\\
2.56410256410256	0.936010066117591\\
2.82051282051282	0.884580311913943\\
3.07692307692308	0.837720368442281\\
3.33333333333333	0.79494383395249\\
3.58974358974359	0.75582081658224\\
3.84615384615385	0.719971036262722\\
4.1025641025641	0.68705779515633\\
4.35897435897436	0.656782704959771\\
4.61538461538461	0.628881073978844\\
4.87179487179487	0.603117869530876\\
5.12820512820513	0.579284182213393\\
5.38461538461539	0.55719412811481\\
5.64102564102564	0.536682133326598\\
5.8974358974359	0.517600552312949\\
6.15384615384615	0.499817577947396\\
6.41025641025641	0.483215406461021\\
6.66666666666667	0.467688625271974\\
6.92307692307692	0.453142795774809\\
7.17948717948718	0.439493206741852\\
7.43589743589744	0.426663777097974\\
7.69230769230769	0.414586089535814\\
7.94871794871795	0.403198538793687\\
8.2051282051282	0.392445580469108\\
8.46153846153846	0.382277068026912\\
8.71794871794872	0.372647667216997\\
8.97435897435897	0.363516338472746\\
9.23076923076923	0.354845879043328\\
9.48717948717949	0.34660251764392\\
9.74358974358974	0.338755555307116\\
10	0.331277046903441\\
};
\addlegendentry{Ach.~\eqref{eq:Last_achiev}, $F_X=R_X$}

\addplot [color=black, dashdotted, thick]
  table[row sep=crcr]{%
0	1.56106951246009\\
0.256410256410256	1.40868607432553\\
0.512820512820513	1.27462837286063\\
0.769230769230769	1.15652333985444\\
1.02564102564103	1.0523170924776\\
1.28205128205128	0.960230997116218\\
1.53846153846154	0.878723876083927\\
1.79487179487179	0.806459488657824\\
2.05128205128205	0.742278541735777\\
2.30769230769231	0.685174591485985\\
2.56410256410256	0.634273288220858\\
2.82051282051282	0.588814494569525\\
3.07692307692308	0.548136873722734\\
3.33333333333333	0.511664601684787\\
3.58974358974359	0.478895906461797\\
3.84615384615385	0.449393179117787\\
4.1025641025641	0.422774437644331\\
4.35897435897436	0.398705955474781\\
4.61538461538461	0.37689589296609\\
4.87179487179487	0.357088792898985\\
5.12820512820513	0.339060820549685\\
5.38461538461539	0.322615645624371\\
5.64102564102564	0.307580877716264\\
5.8974358974359	0.293804979282376\\
6.15384615384615	0.281154590732357\\
6.41025641025641	0.269512211323401\\
6.66666666666667	0.25877418737532\\
6.92307692307692	0.248848966040661\\
7.17948717948718	0.239655578642176\\
7.43589743589744	0.231122322557687\\
7.69230769230769	0.22318561490506\\
7.94871794871795	0.21578899495586\\
8.2051282051282	0.208882255369544\\
8.46153846153846	0.202420685063043\\
8.71794871794872	0.196364408875256\\
8.97435897435897	0.190677811205508\\
9.23076923076923	0.185329032545095\\
9.48717948717949	0.180289529320717\\
9.74358974358974	0.175533688761666\\
10	0.171038491617808\\
};
\addlegendentry{Ach.~\eqref{eq:Second_Last_achiev}, $F_X=R_X$}

\addplot [color=black, thick, dashed]
  table[row sep=crcr]{%
0	0.926678158905707\\
0.256410256410256	0.802029616707351\\
0.512820512820513	0.706302335579836\\
0.769230769230769	0.631313279343153\\
1.02564102564103	0.571386830825779\\
1.28205128205128	0.52255307601459\\
1.53846153846154	0.482010366728699\\
1.79487179487179	0.447762815940995\\
2.05128205128205	0.418374253433097\\
2.30769230769231	0.392800168836712\\
2.56410256410256	0.370272188811423\\
2.82051282051282	0.350218154407906\\
3.07692307692308	0.332206468455839\\
3.33333333333333	0.315907088633669\\
3.58974358974359	0.301064005937702\\
3.84615384615385	0.287475695730514\\
4.1025641025641	0.274981136190271\\
4.35897435897436	0.263449737867867\\
4.61538461538461	0.252774037233913\\
4.87179487179487	0.242864355233422\\
5.12820512820513	0.233644861222729\\
5.38461538461539	0.225050648139088\\
5.64102564102564	0.217025539777214\\
5.8974358974359	0.20952043143087\\
6.15384615384615	0.202492021630473\\
6.41025641025641	0.19590183258808\\
6.66666666666667	0.189715445267575\\
6.92307692307692	0.183901895189809\\
7.17948717948718	0.178433189556822\\
7.43589743589744	0.173283916705273\\
7.69230769230769	0.168430926445099\\
7.94871794871795	0.163853065326729\\
8.2051282051282	0.159530954889386\\
8.46153846153846	0.155446803886352\\
8.71794871794872	0.151584247654161\\
8.97435897435897	0.147928209402217\\
9.23076923076923	0.1444647793985\\
9.48717948717949	0.141181108925176\\
9.74358974358974	0.138065316554419\\
10	0.135106404807126\\
};
\addlegendentry{Ach.~\eqref{eq:Exect_Achiev}, $F_X= R_X$}

\addplot [color=black, dotted, thick]
  table[row sep=crcr]{%
0	0.926678158905707\\
0.256410256410256	0.802029616707351\\
0.512820512820513	0.706302335579836\\
0.769230769230769	0.631313279343153\\
1.02564102564103	0.571386830825779\\
1.28205128205128	0.52255307601459\\
1.53846153846154	0.482010366728699\\
1.79487179487179	0.447762815940995\\
2.05128205128205	0.415351993384161\\
2.30769230769231	0.384064804746933\\
2.56410256410256	0.361478569552778\\
2.82051282051282	0.332156121240441\\
3.07692307692308	0.318547098048289\\
3.33333333333333	0.296316880552112\\
3.58974358974359	0.282823474268607\\
3.84615384615385	0.256196838102942\\
4.1025641025641	0.240468775016587\\
4.35897435897436	0.229782504424074\\
4.61538461538461	0.21732344199738\\
4.87179487179487	0.202425365924364\\
5.12820512820513	0.191648327936008\\
5.38461538461539	0.182688903510854\\
5.64102564102564	0.16682071836134\\
5.8974358974359	0.154514196696357\\
6.15384615384615	0.154398630932603\\
6.41025641025641	0.136508899093823\\
6.66666666666667	0.121484879230668\\
6.92307692307692	0.130912898086112\\
7.17948717948718	0.118492528581412\\
7.43589743589744	0.104072707330257\\
7.69230769230769	0.101404937400595\\
7.94871794871795	0.102777055876082\\
8.2051282051282	0.0957358871383032\\
8.46153846153846	0.090873730739845\\
8.71794871794872	0.0770527942267144\\
8.97435897435897	0.0842326683271192\\
9.23076923076923	0.074876248410419\\
9.48717948717949	0.0730595502961868\\
9.74358974358974	0.0742258227152651\\
10	0.0598234471935863\\
};
\addlegendentry{Ach.~\eqref{eq:Exect_Achiev} w. opt. over $F_X$}

\addplot [color=black, thick]
  table[row sep=crcr]{%
0	0.86897353302468\\
0.256410256410256	0.706141922718349\\
0.512820512820513	0.56982444572697\\
0.769230769230769	0.444294694027503\\
1.02564102564103	0.32972418876405\\
1.28205128205128	0.226466779476781\\
1.53846153846154	0.133405371347466\\
1.79487179487179	0.0495331752802962\\
2.05128205128205	0\\
2.30769230769231	0\\
2.56410256410256	0\\
2.82051282051282	0\\
3.07692307692308	0\\
3.33333333333333	0\\
3.58974358974359	0\\
3.84615384615385	0\\
4.1025641025641	0\\
4.35897435897436	0\\
4.61538461538461	0\\
4.87179487179487	0\\
5.12820512820513	0\\
5.38461538461539	0\\
5.64102564102564	0\\
5.8974358974359	0\\
6.15384615384615	0\\
6.41025641025641	0\\
6.66666666666667	0\\
6.92307692307692	0\\
7.17948717948718	0\\
7.43589743589744	0\\
7.69230769230769	0\\
7.94871794871795	0\\
8.2051282051282	0\\
8.46153846153846	0\\
8.71794871794872	0\\
8.97435897435897	0\\
9.23076923076923	0\\
9.48717948717949	0\\
9.74358974358974	0\\
10	0\\
};
\addlegendentry{Converse in~\eqref{eq:converse}}

\end{axis}
\end{tikzpicture}%
\caption{Bounds for the binomial source in Example~\ref{examp:binomial}, $M\!=\!8$.}
\label{fig:example_M8_vary_g}
\end{figure}
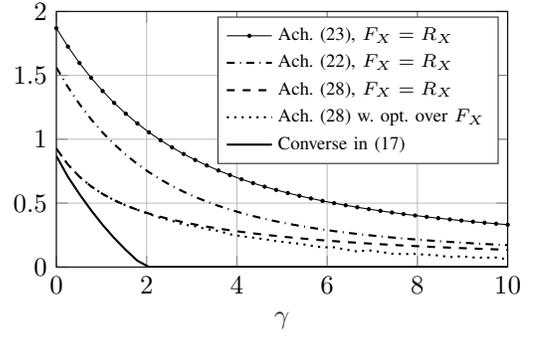

\section{Conclusion}
\label{sec:Conclusions}
This work addresses the problem of lossy source coding with a novel distortion function: the focal loss. The focal loss is a modified version of the log loss, which incorporates an additional term that de-emphasizes examples with higher soft values. In the single-shot setting, this work provides both converse and achievability bounds. In the asymptotic regime, it is shown that the focal term has no impact, and the distortion-rate trade-off for the focal loss aligns with that of the log loss. Numerical examples illustrate the effect of the focal term in the non-asymptotic setting.

Several future directions are worth exploring. First, while the converse is tight in the large blocklength regime, it may require improvement in the non-asymptotic setting. Second, this work focuses on the average distortion, and addressing the excess distortion poses interesting challenges. Third, the multiterminal version (e.g., CEO) of the problem could yield valuable insights. Finally, considering negative values of $\gamma$ is intriguing, as the focal term in such cases prevents soft values from approaching one, introducing confusion among several values - this is a setup with potential privacy applications.

\begin{appendices}

\section{Proof of Proposition~\ref{prop:max_entropy}} 
\label{app:MaxEntrPrinc}
As was already alluded to earlier, the proof of Proposition~\ref{prop:max_entropy} relies on the observation that quantities of the form 
\begin{equation} \log \left( \sum_{x \in \cX} f(P_X(x)) \right),
\end{equation} where $f: [0,1] \to [0,1]$ is continuous and has a single inflection point on $[0,1]$, have a special optimizing structure.  Results of this type can be found in~\cite[Prob.~21]{mildorf_inequalities},\cite{aops_problem_solving,math_stack_max_ent}.

Indeed the function $f_\gamma(x) = x^{(1-x)^\gamma}$ has this structure as stated next\footnote{The proof of Lemma~\ref{lem:convex_concave} is a simple, yet rather tedious, exercise and it is left to the reader.}.
\begin{lemma}\label{lem:convex_concave} Let $x \in  [0,1]$ and $\gamma \geq 0$, and define
\begin{equation}
    f_\gamma(x) = x^{(1-x)^\gamma}. 
\end{equation}
Then, $f_\gamma(x)$ is convex on $[0,c]$ and concave on $[c,1]$ for some $c \in (0,1)$.
\end{lemma}

The following theorem is key to finding the maximum entropy like principle. 
\begin{theorem} \label{thm:maximiziation_theorem}
Let $x_1, x_2, \cdots, x_n$ be $n$ real numbers such that:
\begin{itemize}
    \item $x_1 \leq x_2 \leq \cdots \leq x_n$,
    \item $x_1, x_2, \cdots, x_n \in [a, b]$,
    \item $x_1 + x_2 + \cdots + x_n = C$ \quad (where $C$ is a constant),
\end{itemize}
and let $f$ be a function defined on $[a, b]$, such that $f$ is convex on $[a, c]$ and concave on $[c, b]$. 
Define 
\begin{equation}
\label{eq:FFunc}
F = f(x_1) + f(x_2) + \cdots + f(x_n).
\end{equation}
Then, $F$ is maximal at $x_1 = x_2 = \cdots = x_{k-1} = a$ and $x_{k+1} = \cdots = x_n$, for some $k = 1, 2, \cdots, n$.
\end{theorem}
\begin{IEEEproof}
    Let $\cS$ contain all the $n$-tuples that maximize $F$ in~\eqref{eq:FFunc}. Also, let  
    \begin{equation}
    M(x_1,\ldots, x_n) =  \sum_{i=1}^n 1_{\{ i: a<x_i<c  \} },
    \end{equation}
    and
 \begin{equation}
\kappa = \max_{(x_1,\ldots, x_n) \in \cS} M(x_1,\ldots, x_n) . 
\end{equation}
We 
%want to show 
start by showing that $\kappa \le 1$. Towards a contradiction, assume that $\kappa \ge 2$. This implies that there exists $x_i, x_j \in (a,c)$ with $i <j$. 
We now distinguish two cases:
\begin{enumerate}
\item $x_i+x_j-c\ge a$: using the majorization inequality, we have that 
\begin{equation}
f(x_i) + f(x_j) \le f(x_i+x_j -c) + f(c). 
\end{equation}
This implies that   $x_i+x_j -c, c $ can be substituted for $x_i,x_j$ and lead to a 
strictly 
smaller value of $\kappa$. This results in a contradiction. 
\item  $x_i+x_j-c < a$: using the majorization inequality, we have that 
\begin{equation}
f(x_i) + f(x_j) \le f(x_i+x_j -a) + f(a) .
\end{equation}
This implies that   $x_i+x_j -a, a $ can be substituted for $x_i,x_j$ and lead to a strictly smaller value of $\kappa$. This results in a contradiction.
\end{enumerate}
Since both cases above result in a contradiction, we have shown that $\kappa \le 1$.  

The above implies that there exists some $k \in [1:n]$ such that $x_1= \ldots= x_{k-1} =a$ and $c \le x_{k+1} \le \ldots \le x_n\le b$. To conclude the proof, note that by Jensen's inequality and the concavity of $f$ on $[c,b]$, we have that
\begin{equation}
    f(x_{k+1})+ \ldots + f(x_n) \le (n-k)  f \left(  \frac{x_{k+1} + \ldots + x_n}{n-k}\right), 
\end{equation}
which holds with equality if $x_{k+1} = \ldots =x_n$. This concludes the proof of Theorem~\ref{thm:maximiziation_theorem}. 
\end{IEEEproof}

Using Theorem~\ref{thm:maximiziation_theorem} and Lemma~\ref{lem:convex_concave}, we have that the maximizer $P_{X^\star}$ of
\begin{equation}
    \max_{P_X : X \in \cX}  H_\gamma(P_X) \label{eq:intermidiate_step} 
\end{equation}
satisfies the following properties: for some $k \in [1: | \cX | ]$
\begin{itemize}
    \item $P_{X^\star}(x_i) =0  ,  \, i \in [1: k-1]$;
    \item $P_{X^\star}(x_k) = p $; 
    \item $P_{X^\star}(x_i) = \frac{1-p}{| \cX| -k}  , \,  i \in [k+1, | \cX| ]$.
\end{itemize}
Clearly, $k=| \cX|$, i.e., $P_X(x_{|\cX|})=1$ and $P_X(x_i)=0, i \in [1:|\cX|-1]$, is not a maximizer since it would lead to $H_\gamma(P_X)=0$. This implies that the maximization in~\eqref{eq:intermidiate_step} is equivalent to~\eqref{eq:expression_for_max}, at the top of the next page.
This proves~\eqref{eq:FirstClaimPropMaxEntr} i.e., the first claim in Proposition~\ref{prop:max_entropy}.
\begin{figure*}[h!]
\begin{align}
& \max_{ \small \begin{array}{c} k \in [1: | \cX|-1] \\  p \in [0,1] \end{array}} \log \left( p^{(1-p)^\gamma} + (| \cX| -k)  \left( \frac{1-p}{|\cX| -k} \right)^{ \left( 1 - \frac{1-p}{| \cX| -k} \right)^\gamma } \right) \notag\\
 &=  \max_{ \small \begin{array}{c} d \in [1: | \cX|-1] \\  p \in [0,1] \end{array}} \log \left( p^{(1-p)^\gamma} + d \left( \frac{1-p}{d} \right)^{ \left( 1 - \frac{1-p}{d} \right)^\gamma } \right). \label{eq:expression_for_max}
\end{align}
\hrule
\end{figure*}

To show~\eqref{eq:SecondClaimPropMaxEntr}, i.e., the upper bound in Proposition~\ref{prop:max_entropy}, we note that  
\begin{align}
d \left( \frac{1}{d} \right)^{ \left( 1 - \frac{1-p}{d} \right)^\gamma } &= d^{ 1 -  \left( 1 - \frac{1-p}{d} \right)^\gamma }  \notag \\
&\stackrel{{\rm{(a)}}}{\leq} d^{ \max(1,\gamma) \frac{1-p}{d}} 
\notag
%\label{eq:max_over_d} 
\\
&\stackrel{{\rm{(b)}}}{\leq}  \rme^{ \max(1,\gamma) \frac{1-p}{\rme}} 
\notag%\label{eq:max_over_d_second_time}
\\
&\stackrel{{\rm{(c)}}}{\leq}  \rme^{ \frac{\max(1,\gamma) }{\rme}} ,\label{eq:bound_on_p}
\end{align}
where the labeled inequalities follow from:
$\rm{(a)}$ using the fact that 
\begin{align}
1 -  \left( 1 - \frac{1-p}{d} \right)^\gamma   \le 
 \left \{ \begin{array}{cc}    \frac{ \gamma (1-p)}{d}  & \gamma \ge 1,  \\
  \frac{1-p}{d}& \gamma \in (0,1), \end{array} \right.   
\end{align} 
where for $\gamma \ge 1$ we have used the Bernoulli inequality  (i.e., $(1+x)^\gamma  \ge 1+ \gamma x,\, \gamma \ge 1, \, x \ge -1$ ), and for $\gamma \in (0,1)$ we have used the fact that $\left( 1 - \frac{1-p}{d} \right)^\gamma  \ge \left( 1 - \frac{1-p}{d} \right)$;
%~\eqref{eq:max_over_d_second_time} follows by 
$\rm{(b)}$ maximizing over $d>0$; 
and $\rm{(c)}$
%~\eqref{eq:bound_on_p} is due to 
the bound $1-p \le 1$.

Now, combining~\eqref{eq:expression_for_max} and~\eqref{eq:bound_on_p}, and using the fact that $p^{(1-p)^\gamma} \le 1 $ and  $\left( 1-p \right)^{ \left( 1 - \frac{1-p}{d} \right)^\gamma } \le 1 $ we arrive at the following bound
\begin{align}
&\max_{ \small \begin{array}{c} d \in [1: | \cX|-1] \\  p \in [0,1] \end{array}} \log \left( p^{(1-p)^\gamma} + d \left( \frac{1-p}{d} \right)^{ \left( 1 - \frac{1-p}{d} \right)^\gamma } \right) \notag\\
& \le \log \left( 1 + \rme^{ \frac{\max(1,\gamma) }{\rme}}  \right) .
\end{align}
This concludes the proof of Proposition~\ref{prop:max_entropy}.

\end{appendices}

\bibliographystyle{ieeetr}
\bibliography{biblio.bib}
\end{document}